\documentclass[11pt]{article}

\usepackage[english]{babel}

\usepackage{cite}

\usepackage[letterpaper,top=1in,bottom=1in,left=1in,right=1in]{geometry}

\usepackage{amsmath}
\usepackage{amssymb}
\usepackage{amsthm}
\usepackage{graphicx}
\usepackage{mathtools}
\usepackage{xcolor}
\usepackage[colorlinks=true, allcolors=blue]{hyperref}

\newcommand{\pagebreakk}{}

\newtheorem{theorem}{Theorem}
\newtheorem{lemma}[theorem]{Lemma}
\newtheorem{defn}[theorem]{Definition}

\newtheorem{cor}[theorem]{Corollary}
\newtheorem{obs}[theorem]{Observation}

\newcommand{\maxedge}{\text{\rm maxedge}}

\newcommand{\jlabel}[1]{\label{#1}}

\newcommand{\ballcover}{\text{\rm BallCover}}

\DeclareMathOperator*{\near}{near}

\newcommand{\finalfigfrompage}[2]{%
\begin{center}%
\includegraphics[width=4in,page=#1,trim=0 #2 0 0,clip]{ipefigs.pdf}%
\end{center}}

\title{Distances and shortest paths on graphs\\ of bounded highway dimension: \\ simple, fast, dynamic}
\author{Sébastien Collette\thanks{Synapsis Group; work on \href{syty.io}{syty.io} supported and subsidized by \href{Innoviris.brussels}{Innoviris.brussels}. {\tt sco@synapsis-group.com}} \and John Iacono\thanks{Université Libre de Bruxelles \& Synapsis Group; work on \href{syty.io}{syty.io} supported and subsidized by \href{Innoviris.brussels}{Innoviris.brussels}. {\tt john.iacono@ulb.be}}}
\date{}

\begin{document}
\maketitle

\begin{abstract}
Dijkstra's algorithm is the standard method for computing shortest paths on arbitrary graphs. 
However, it is slow for large graphs, taking at least linear time.
It has been long known that for real world road networks, creating a hierarchy of well-chosen shortcuts allows fast distance and path computation, with exact distance queries seemingly being answered in logarithmic time. However, these methods were but heuristics until the work of Abraham et al.~[JACM 2016], where they defined a graph parameter called highway dimension which is constant for real-world road networks, and showed that in graphs of constant highway dimension, a shortcut hierarchy exists that guarantees shortest distance computation takes $O(\log (U+V))$ time and $O(V \log (U+V))$ space, where $U$ is the ratio of the smallest to largest edge, and $V$ is the number of vertices. The problem is that they were unable to efficiently compute the hierarchy of shortcuts. Here we present a simple and efficient algorithm to compute the needed hierarchy of shortcuts in time and space $O(V \log (U+V))$, as well as supporting updates in time $O( \log (U+V))$.
\end{abstract}

\pagebreakk

\section{Introduction}

\hfill \emph{Aut viam inveniam aut faciam}

\hfill \emph{---Hannibal}
\\ \

The shortest path and shortest distance problems are fundamental problems in computer science: 
in the shortest path problem, given an origin and a destination, find the shortest path between them on the graph. For the shortest distance problem one is interested in the length of the shortest path and not the path itself.

The classical solution for general graphs $G=(V,E)$ is Dijkstra's algorithm from the 1950's \cite{DBLP:journals/nm/Dijkstra59} which, when combined with a heap supporting constant time decrease-key operations, such as a Fibonacci heap \cite{DBLP:journals/jacm/FredmanT87}, can compute shortest paths and distances from any single vertex to all other vertices in the graph in time $O(|V| \log |V| + |E|)$.
This algorithm, by computing all paths and distances from a single vertex, appears to be much too big a hammer when all one is interested in is the shortest distance between a single pair of vertices, but yet Dijkstra's is the best known for general graphs unless one is willing to use superquadratic space and preprocessing by, for example, computing all-pairs shortest paths (e.g.~\cite{DBLP:conf/soda/PettieR02}).
Note that here we are only interested in exact distances, not approximate ones, which are the subject of the distance oracles literature which began with \cite{DBLP:journals/jacm/ThorupZ05}.

The slowness of Dijkstra's algorithm is real: to compute a shortest distance on the road map of the world's $136^{th}$ largest country, Belgium, takes half a second on standard hardware. Computing a single shortest distance query for each Belgian on the map of Belgium would take months of computation\footnote{This is a strong statement to make the reader feel how slow Dijkstra is, we can of course do better with parallel processing, clustering, etc. But even on a 30-core processor, computing millions of Dijkstra instances will take days.}, but as the number of Belgians is roughly the same as the graph size, the number of queries is nowhere near enough to make an all-pairs approach worthwhile, especially if the graph changes.

However, road graphs are not general graphs, and over time a number of heuristics have emerged to compute shortest distances and paths radically faster than using Dijkstra on real-world graphs
\cite{
DBLP:conf/dimacs/Lauther06,
DBLP:conf/atmos/BauerDW07,
DBLP:conf/soda/GoldbergH05,
DBLP:conf/alenex/HolzerSW06,
DBLP:conf/wea/DellingW13,
DBLP:journals/jea/SandersS12,
DBLP:conf/alenex/Gutman04,
DBLP:conf/dimacs/BastFM06,
DBLP:journals/transci/GeisbergerSSV12,
DBLP:conf/esa/AbrahamDGW12,
DBLP:conf/wea/AbrahamDGW11} which largely involve in one form or another of adding pre-computed shortcut edges to the graph.
Consider the following method to create a hierarchy of graphs:

\begin{itemize}
\item Define a sequence of sets of vertices $V=C[0] \supset C[1] \supset C[2] \cdots $.
\item The number of vertices at level at least $i$, $C[\geq i]$ should be exponentially decreasing in $i$: $|C[i]| = \Theta( \frac{|V|}{c^i}) $.
\item Construct shortcut graphs $G[i]$, which are graphs that can be used to compute shortest paths among those vertices in $C[i]$; $G[0]=|V|$. See Figure~\ref{fig:11}.
\end{itemize}

\begin{figure}
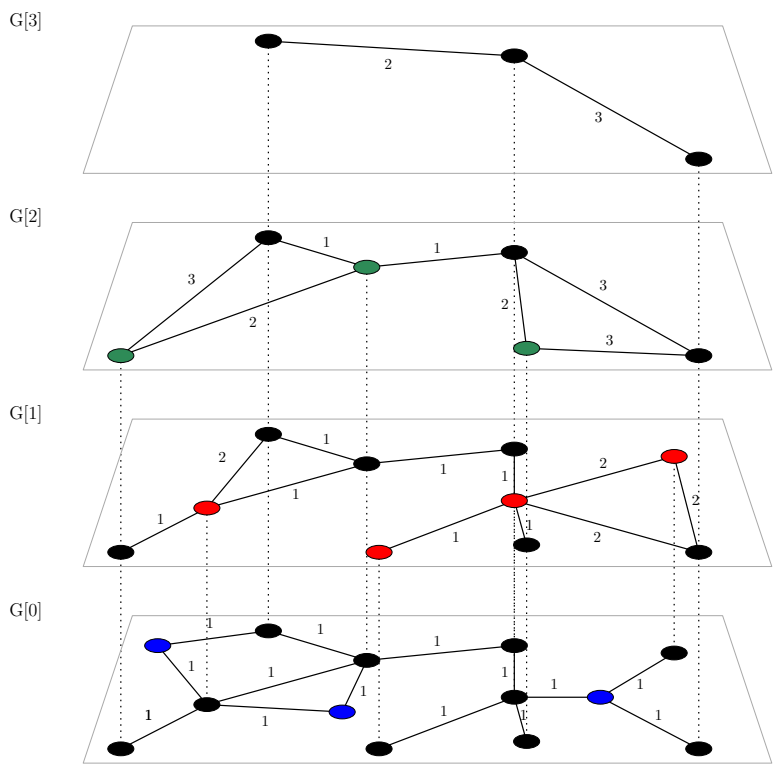

    \centering
    \finalfigfrompage{11}{0}
    \caption{A shortcut graph hierarchy, where the colored vertices are removed to create the next level.}
    \label{fig:11}
\end{figure}

From such a basic structure, which one can think of as a kind of skip list on a graph, one can search for shortest paths that are bitonic, that is from the origin and destination they start at level 0, and only consider paths that continue on the same level or continue to a higher level until the two sides meet in the middle. If such a hierarchy is constructed with nice properties, one can imagine that logarithmic shortest distance queries are envisionable. Such nice properties would involve only needing to follow a constant number of vertices at any level before being assured that one will hit a higher-level vertex, and keeping the degree of the vertices at each level constant. The first condition is easier than the second to obtain, as a random sample will achieve the first condition in expectation but utterly fail in maintaining constant degree even when this is possible though a more considered choice strategy.

How does one compute the shortcut graph $G[i]$ from $G[i-1]$? Start with $G[i]$ and incrementally remove every vertex $v$ in $C[i] \setminus C[i-1]$; when doing so look at all pairs of neighbors $v_1,v_2$ of $v$; if the shortest path from $v_1$ to $v_2$ does not go through $v$ do nothing, otherwise add an edge from $v_1$ to $v_2$ weighted with the sum of the edges $v_1,v$ and $v,v_2$.

\begin{figure}
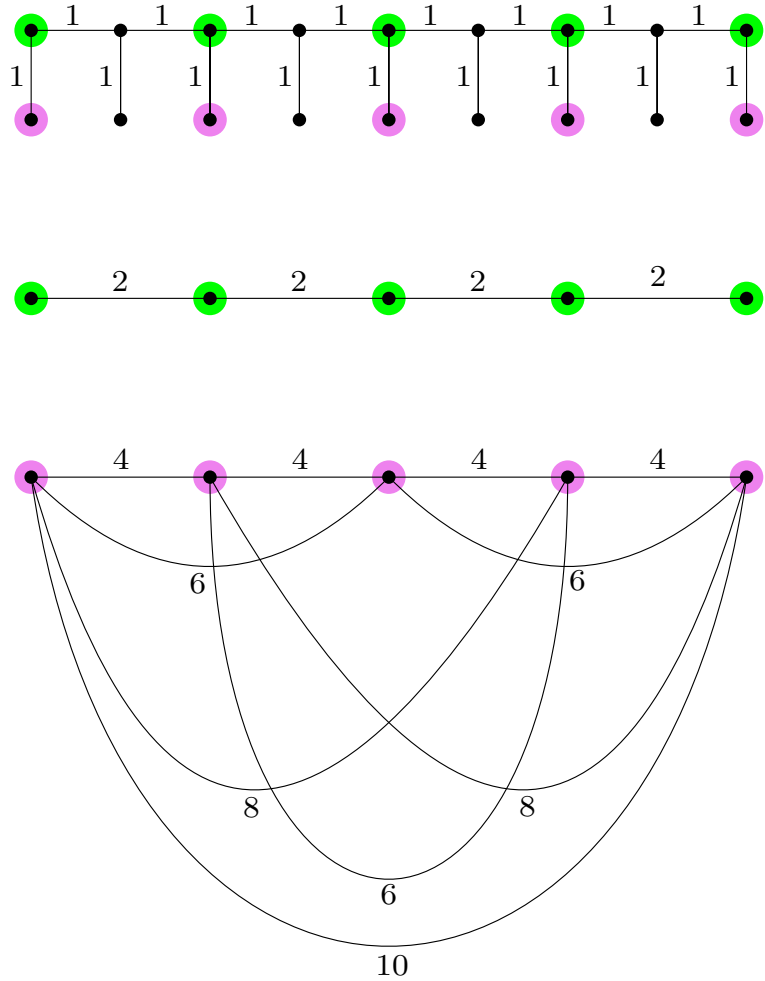

    \centering
    \finalfigfrompage{9}{0}
    \caption{There are bad choices when deciding which vertices should remain in the next shortcut graph. Given the top graph, the shortcut graph of the green vertices is the middle graph, and the magenta vertices is the bottom figure. The middle graph is a path and the bottom a clique.}
    \label{fig:9}
\end{figure}

Now, the crucial point is that the choice of which vertices should be chosen to be in the next level can not be made arbitrarily or randomly. There are bad choices. As illustrated in Figure~\ref{fig:9}, one choice leads to a quadratic-sized, linear degree, and thus useless, shortcut graph and the other leads to a nice-looking linear-sized constant-degree graph.

This figure illustrates the main heuristic used in creating these hierarchies: vertices on more important roads should appear higher up in the hierarchy, and unimportant roads of limited use such as dead ends should remain in low levels only. How can one determine which roads are important and which are not? One can get the information from the map metadata or another source, or attempt to use real-world traffic data to figure it out. However, while this works well, it remains but a heuristic. Also, deciding on the importance of roads in a fixed manner inhibits flexibility and the ability to react dynamically, a gravel road may become vitally important if it can be used to circumnavigate a traffic accident.

In \cite{HD}, heretoforth referred to as the \emph{HD paper}, a major break-though was made. They defined a graph parameter which they called highway dimension, argued that real-world road networks have small highway dimension, and showed that for graphs of small highway dimension there exists a hierarchy of $C[i]$s such that distance queries can be computed in logarithmic time, and shortest path queries take an additional linear term in the path description complexity. 
The formal definition of highway dimension is presented in Definition~\ref{d:hd}, but informally a graph has constant highway dimension, if for every $v$ and $r$, the set of shortest paths of length $r$ that pass close (within $2r$) to $v$ can be covered with a constant number of vertices. Highway dimension has no relation to planarity other than both have a linear number of edges, and constant highway dimension is stronger than constant doubling dimension. Highway dimension neither needs an embedding of the graph nor restricts edge weights to be some distance such as Euclidean.

However, the only problem in the HD paper was they showed only the existence of $C[i]$s that work well, given constant highway dimension. No efficient algorithm to compute them was proposed. We quote from the HD paper, with brackets being our attempt to link their notation to the discussion so far:

\begin{quote}
Although Theorem 4.2 guarantees the existence of good hitting sets [the $C[i]$'s], no polynomial-time algorithm is known for computing the minimum hitting set $H$ used in the proof. As a result, no polynomial-time algorithm is known for computing an optimum multiscale SPHS [hierarchy of $C[i]$s], which are important building blocks of the preprocessing phase of the algorithms we analyze in Sections 5 and 6 [shortcut based algorithms]. Note that in these sections, we analyze the query times assuming optimum (exponential-time) preprocessing.

Section 8 discusses polynomial-time approximation algorithms for SPHS computation, which enable polynomial-time preprocessing with a slight degradation of the query bounds. Although still impractical for large road networks, these polynomial-time preprocessing routines provide some justification for practical variants of preprocessing. These variants run in polynomial time and produce solutions that have no theoretical quality guarantee, but work well on real-world problems.
\end{quote}

Here we show how to compute $C[i]$s that are proven to work efficiently for graphs of constant highway dimension. Our idea, which we motivate further in Section \ref{s:appraoch}, is that a vertex is on an important road if it is in the middle of a long shortest path. We show that this heuristic is theoretically sound and algorithmically efficient. As a bonus, we show that we can not only compute the $C[i]$s efficiently, we show that we can maintain them under edge insertions, deletions and re-weightings. We summarize our result with the following theorem:

\begin{theorem}
Let $G$ be a connected graph with $|V|$ vertices, constant highway dimension, and ratio of smallest-to-largest edge $U$. Our data structure supports
shortest distance queries in time $O(\log (|V|+U))$ and shortest path queries whose result is a path with $p$ edges in time $O(p+\log (|V|+U))$. Our data structure can be computed in $O(|V| \log U)$ time and takes $O(|V| \log U)$ space. The data structure supports dynamic changes to the graph, edge insertions/deletions/reweighting in time $O(\log (|V|+U))$.
\end{theorem}

We note that the space usage is the same as in the HD paper and is linear for unit-edge-weight graphs. In practice, for open street map data for Belgium, $U\approx 0.6|V|$. With some additional analysis, the space usage of our structure should be able to be shown to be $$O\left(\sum_{v\in V} \log \frac{\text{weight of largest incident edge to $v$}}{\text{weight of smallest incident edge to $v$}}\right);$$ for the Open Street Map data for Belgium this would be $1.8|V|$ inside the $O()$ when using the binary logarithm.

\paragraph{Other work on highway dimension}

Here we briefly review the literature that uses the notion of highway dimension. This literature has a similar flavor to treewidth in that a number of problems have been shown to be more efficiently solvable on graphs of constant highway dimension than on more general graph classes.

This includes:
capacitated vehicle routing \cite{DBLP:journals/talg/JayaprakashS23},
generalized $k$-center
\cite{DBLP:conf/wg/FeldmannV22},
the traveling salesperson problem
\cite{DBLP:journals/algorithmica/DisserFKK21},
clustering \cite{DBLP:journals/jcss/FeldmannS21},
$k$-center problems \cite{DBLP:journals/algorithmica/Feldmann19}, and
embedding of low highway dimension graphs into graph of bounded treewidth
\cite{DBLP:journals/siamcomp/FeldmannFKP18}.
In \cite{DBLP:conf/icalp/0002KV19}, they show that graphs of constant highway dimension have exponentially smaller 3-hopsets than previously known distance oracles (a $k$-hopset is a set of edges added to a graph the does not change shortest path lengths but shortest paths can be computed only using $k$ edges).
In \cite{DBLP:conf/iwpec/0001DF0Z22} it is shown that finding the smallest set of vertices that intersect all paths of a certain length in a graph of constant highway dimension is $W[1]$ hard.
In \cite{DBLP:conf/soda/KosowskiV17}, they introduce a new graph parameter, \emph{skeleton dimension}, and show that hub labeling schemes in graph with low skeleton dimension are more efficient than for those known for graphs with small highway dimension.

\pagebreakk\section{Our approach}
\jlabel{s:appraoch}

Here we describe at a high level our approach. We make one assumption in this section that greatly simplifies our arguments: all edges in $G$ have unit length, perturbed in such a way that shortest paths are unique. Much of the complexity in the proofs in the main part of the paper has to do with making the definitions and proofs work with variable length edges, but our main intellectual contribution can be presented limiting the discussion to perturbed unit-length edges. We give forward pointers to the related claims in the main presentation, keeping in mind that the exact statements often do not match due to the complexities of handling variable-length edges.

We create a hierarchy of graphs called \emph{shortcut graphs $G[i]$ (Definition~\ref{d:shortcut})}  and vertex sets called \emph{vertex covers} $C[i]$ (Definition~\ref{d:vc}) , where $C[i]$ are the vertices of $G[i]$. The base case has $G[0]$ being $G$ and $C[0]$ being $V$. We construct $C[i]$ from $G[i-1]$ incrementally by initializing $C[i]$ to be empty, and then looking at all paths in $G[i]$ of length between $\frac{3}{4}8^i$ and $8^i$, and if these paths currently do not contain a vertex in $C[i]$ add the vertex in $G[i]$ closest to the midpoint of the path to $C[i]$. The shortcut graph $G[i]$ has edges between vertices of $C[i]$ who do not contain another vertex of $C[i]$ on their shortest path; these edges are weighted with the shortest-path distance. We call this the pick-a-vertex-close-to-the-middle-if-there-is-no-vertex-on-the-path-yet method, the full version is in Definition~\ref{d:construct}.

This hierarchy has a number of nice properties, assuming constant highway dimension. 
\begin{itemize}
    \item Shortest path lengths between vertices of $C[i]$ in $G[i]$ are the same as in $G$ (Observation~\ref{o:shortcutfacts}).
    \item Vertex degree in all $G[i]$ is constant (Lemma~\ref{l:degree}).
    \item Any shortest path of length at least $8^i$ in $G$ contains at least one element of $C[i]$ (Lemma~\ref{l:ciscover}).
    \item Edge length in $G[i]$ is at most $8^i$ (by the previous point, and the meaning of a shortcut graph).
    \item The size of the part of $G[i]$ graph within distance $8^i$ of any element of $C[i]$ is constant (Lemma~\ref{l:shorcutsize}).
    \item The total number of levels is $O(\log |V|)$ (Lemma~\ref{l:height}).
    \item The total size of all $G[i]$ and $C[i]$ is $O(|V|)$ (Lemma~\ref{l:size}).
\end{itemize}

These properties ensure that the creation of the $C[i]$'s and $G[i]$ can be done locally by only examining constant-size pieces of $C[i-1]$ and $G[i-1]$ around the parts being created, which gives linear total construction time (Section~\ref{s:computing}). It also gives logarithmic update time, since only a constant sized piece of each $G[i]$ and $C[i]$ needs to be recomputed around any change (Section~\ref{s:dynamic}).

Given this structure, a variant of bidirectional Dijkstra can be used to compute the distance between two vertices using the union of the $G[i]$s proceeding one level at a time, and it can been shown that the shortest distance can be found using a path which is bitonic in which level of $G[i]$ each edge comes from, and only needs to follow a constant number of edges in each $G[i]$. This ensures that the shortest-path distance can be computed in time $O(\log d)$, where $d$ is the distance. 
Through appropriate augmentation, the shortest path can by obtained in linear time in its description by uncompressing the edges in the shortcut graph. 
Everything described except how the $C[i]$s are computed was already shown in the HD paper.

The pick-a-vertex-close-to-the-middle-if-there-is-no-vertex-on-the-path-yet method is the main new idea, and proving that it works requires one key observation. It was known that if a graph has constant highway dimension, then for any $r$ there is a set of points $C$ such that every path of length at least $r$ has an element of $C$ and within distance $r$ of any vertex there are at most a constant number of elements of $C$ (Lemma~\ref{l:hdchd}). We extend this idea to show that there is a set of paths $P$ of length at least $\frac{1}{4}r$ such that each path $p$ of length at least $r$ contains as a subpath a path of $P$ that includes the middle of $p$, and within distance $r$ of any vertex there are at most a constant number of paths of $P$ (Lemma~\ref{l:hdpc}). Intuitively, we show that all long paths not only pass through a point in a locally sparse point set $C$, they pass completely through a path in a locally sparse set of long paths; this makes sense as if one posits long paths all pass through some point on a highway, they are not on the highway just at one point, but typically for a substantial distance. We do not compute this path cover $P$, but use its existence, through a charging argument, to show that the pick-a-vertex-close-to-the-middle-if-there-is-no-vertex-on-the-path-yet method generates a valid cover which has the required local sparseness (Lemma~\ref{l:ciscover}).

We emphasize that the pick-a-vertex-close-to-the-middle-if-there-is-no-vertex-on-the-path-yet method is simple and efficient, and the resultant hierarchy of vertex covers which gives efficient shortest path computation is exactly what was shown to exist in the HD paper but was unable to be computed.

\pagebreakk
\section{Notation and definitions}

In this section we review the definitions surrounding highway dimension, and introduce our concept of a path cover which is the vital new ingredient which is key to making our algorithm work.

\subsection{Notation.}

The input graph $G$ is a weighted connected undirected graph. We require that all shortest paths are unique, and the shortest path between two vertices which are connected by the edge be that edge. (The latter, which we call the requirement that all edges be useful, is required for Lemma~\ref{lem:bigedge} and is in the HD paper as well). For convenience, we require that all edges have at least unit weight. We will also assume that the graph $G$ has constant highway dimension, as defined in Definition~\ref{d:journalhd}; many other lemmas make claims about other quantities being constant which are dependent on this assumption\footnote{The astute reader will notice that all of our claims of $O(1)$ if highway dimension is constant hide a polynomial dependence on the highway dimension; thus if highway dimension were to be logarithmic, these constants would become polylogarithmic in highway dimension. We have chosen to use $O(1)$ in order to reduce the clutter needed to work out the exact dependence though each lemma and thus improve the readability of the presentation. We also have made no attempts to minimize this polynomial dependence, always opting for the cleanest argument.}. 

We begin with basic notation:
\begin{itemize}

    \item $U_G$, the maximum weight edge in $G$
    \item $p_G(v_1,v_2)$. The shortest path from $v_1$ to $v_2$ in $G$, which we have assumed is uniquely determined; Thus if $v_3,v_4$ are on $p_G(v_1,v_2)$, $p_G(v_3,v_4)$ is a subpath of $p_G(v_1,v_2)$.
    \item $d_G(v_v,v_2)$. The distance between vertices is defined as the length of the shortest path, that is, the sum of the weights of each edge.
    \item $\maxedge(p)$. The maximum weight edge on path $p$.
    \item $B_G(v,r) \coloneqq \{v'|d_G(v,v')\leq r\} $. The ball of radius $r$ from $v$, expressed as a set of vertices.
    \item $V(\cdot)$ the vertices of whatever is in the parenthesis: a graph, a path, a set of edges.
    \item $|e|$ is the weight of edge $e$ ; we use interchangeably the terms length and weight.
\end{itemize}

We will occasionally need the continuous view of a weighted graph, where each edge can be viewed as having an infinity of vertices at each distance along the edge: We will always use a hat to denote a vertex that conceptually may be on an edge. We use $\hat{V}_G$ to denote all possible vertices including those on edges. A path $\hat{p}$ in $G$ is the set of all elements of $\hat{V}_G$ in the path, and the ball $\hat{B}_G(\hat{v},r)$ is the union of the paths of length 
$r$ from $\hat{v}$ to elements of $\hat{V}_G$ in $G$.

We omit the $G$ subscripts if there is no risk of ambiguity that $G$ is the input graph.

\subsection{Several flavors of highway dimension.}

Here we present three different variants of the definition of highway dimension. The first, continuous highway dimension, 
was central in the conference version of the HD paper, but is dependent on the continuous view of graphs. The second definition, discrete highway dimension is new here, and allows simple arguments in many proofs to come. The third, highway dimension without an adjective, comes from the journal version of highway dimension, and we present it last as it is the least intuitive.

We show in Section \ref{s:hdf} that graphs of constant highway dimension have constant continuous  and discrete highway dimension (Corollary~\ref{cor:hd}). Thus, assuming constant highway dimension we can use whichever of these definitions is the easiest to work with, which is often the discrete highway dimension.

We begin with the definition central to the conference version of the paper on highway dimension \cite{HDC} which is elegant and easy to understand: 

\begin{defn}\jlabel{d:hd}{\bf Continuous highway dimension, Definition 11.1 of \cite{HD}}
The continuous highway dimension of a weighted graph $G=(E,V)$ is the minimum ${h}$ such that for every $\hat{v}\in \hat{V}$ and every positive $r$, there exists a set $\hat{C}\subseteq\hat{V}$ of at most ${h}$ vertices such that for every $\hat{v}_1,\hat{v}_2 \in \hat{V}$, if $\hat{p}(\hat{v}_1,\hat{v}_2) \subseteq \hat{B}(\hat{v},4r)$ and $d(\hat{v}_1,\hat{v}_2) > r$ then $\hat{p}(\hat{v}_1,\hat{v}_2) \cap \hat{C} \not = \emptyset$.
\end{defn}

Next is our strictly weaker version that we call \emph{discrete highway dimension} that we will find of use, it simply repeats the continuous definition but restricts the vertices under consideration to be vertices of $V$. In Lemma~\ref{l:dhdhd} we show that a graph of continuous highway dimension $h$ has discrete highway dimension at most $h$. However, the converse does not hold as, for example, a single vertex with $k$ incident edges of geometrically increasing weights has constant discrete highway dimension but $\Theta(k)$ continuous highway dimension.

\begin{defn}\jlabel{d:dhd} {\bf Discrete highway dimension.}
The discrete highway dimension of a weighted graph $G=(E,V)$ is the minimum $h$ such that for every $v \in V$ and every positive $r$, there exists a set $C\subseteq V$ of at most $h$ vertices such that for every $v_1,v_2 \in V$, if $V(P(v_1,v_2)) \subseteq B(v,4r)$ and $d(v_1,v_2)> r$ then $V(p(v_1,v_2))\cap C \not = \emptyset$.
\end{defn}

Finally we include the definition of highway dimension used in the journal version, \cite{HD}, which we wave been referring to as the HD paper. We include this last as in our opinion it is the least intuitive and requires first introducing and understanding some additional notation. It was shown in the HD paper with a lemma we restate here as Lemma~\ref{l:hdchd} that continuous highway dimension and highway dimension are constant-factor equivalent; in this way this definition is a stronger discrete version of highway dimension than our discrete highway dimension.

Before giving the definition of highway dimension, we introduce the needed notation:

\begin{itemize}
	\item Given a path $p$ from $v_1$ to $v_2$, an $r$-witness of $p$ is a shortest path $p'$ of length at least $r$ from $v_1$ or a neighbor of $v_1$ to $v_2$ or a neighbor of $v_2$ that contains $p$ as a subpath.
	\item A shortest path $p$ is $r$-significant if it has a $r$-witness.
	\item $\mathcal{P}(r)$ is the set of all $r$-significant paths.
	\item Given a path $p$ and a vertex $v$ the distance $d(v,p)$ is the shortest distance from $v$ to any vertex of $p$.
	\item A path $p$ is $(r,d)$ close to $v$ if it is $r$-significant, and has a $r$-witness path $p'$ such that $d(v,p')\leq d$.
	\item Let $S(v,r)$ be the set of all paths that are $(r,2r)$ close to $v$.
\end{itemize}

\begin{defn}\jlabel{d:journalhd}{\bf Highway dimension, Definition~3.4 of~\cite{HD}.}
The highway dimension of a weighted graph $G=(E,V)$ is the smallest $h$ such that for all $r>0$ and $v\in V$, there is a set $C \subseteq V$ of at most $h$ vertices such that for all $p \in S(v,r)$, $V(p) \cap C \not= \emptyset$.
\end{defn}

\subsection{Definitions}


We begin with the notion of a vertex cover, the novel part of our definition compared to previous work is that paths with long edges are not required to be covered. We will show in Lemma~\ref{l:hdpvc} that graphs of constant highway dimension have vertex covers.

\begin{defn} \jlabel{d:vc} {\bf Vertex cover.}
A $(r,k)$ vertex cover of $G=(V,E)$ is a set of vertices $C$ such that 
\begin{enumerate}
\item 
all shortest paths $p(v_1,v_2)$, $v_1,v_2 \in V$, of length $d(v_1,v_2) > r $ and $\maxedge(p(v_1,v_2))
\leq r$ have  some $v_c\in C$ where $v_c \in V(p(v_1,v_2))$ and 
\item 
for all $v \in V$ $|B(v,2r) \cap C|\leq k$.
\end{enumerate}
\end{defn}

Figure~\ref{fig:1} illustrates the vertex cover definition.

\begin{figure}
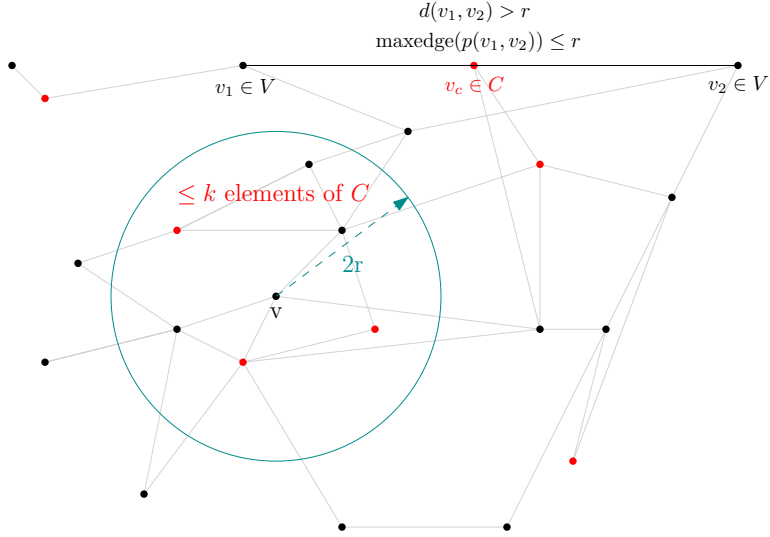

    \centering
    \finalfigfrompage{1}{0}
    \caption{Illustration of a vertex cover in red.}
    \label{fig:1}
\end{figure}


The notion of a path cover expressed here is novel. We will show in Lemma~\ref{l:hdpc} that graphs of constant highway dimension have path covers.

\begin{defn} {\bf Path cover.} \jlabel{d:pc}
A $(r,\eta)$ path cover of $G=(V,E)$,
is a set of shortest paths $P$ between vertices in $V$, each of length from $\frac{1}{4}r$ to $r$ such that 
for any $v_1,v_2\in V$ where the distance $
\frac{1}{2}r\leq d(v_1,v_2)\leq r$ and $\maxedge(p(v_1,v_2)) \leq \frac{1}{8}r$, $p(v_1,v_2)$ contains as a subpath some path $p(v_3,v_4) \in P$, where 
$d(v_1,v_3)$ and $d(v_4,v_2)$ are both at most $\frac{1}{8}r$.

Furthermore, any ball of radius $4 r$ from any $v \in V$ contains at least one vertex of at most $\eta$  paths in $P$.
\end{defn}

Figure~\ref{fig:2} illustrates the definition of a path cover.

\begin{figure}
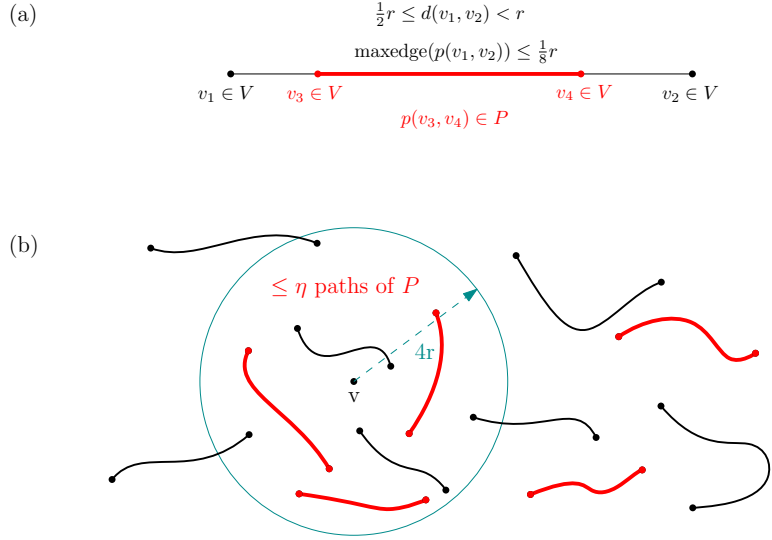

    \centering
    \finalfigfrompage{2}{0}
    \caption{Illustration of a the definition of a path cover. The red paths are the path cover.}
    \label{fig:2}
\end{figure}

Here we present our definition of a shortcut graph. As in the definition of a vertex cover, we give an upper bound on the length of a path for a shortcut edge to appear, this differs from previous work.

\begin{defn}\jlabel{d:shortcut}
 {\bf Shortcut graph, see Section~6.1 of \cite{HD}.}
Given $G=(V,E)$ and $C\subseteq V$, the $r$-shortcut graph of $C$ is $G(C,r)=(C,E')$ where there is an edge from $v_1$ to $v_2$ in $E'$ of weight $d$ iff $d=d_G(v_1,v_2)$, $d\leq r$, and $p_G(v_1,v_2)$ does not contain any elements of $C$ other than $v_1$ and $v_2$.
\end{defn}

The definition of sparse shortest path hitting set (SPHS) is from the HD paper \cite{HD}. In spirit the SPHS, the SPC from the conference version of the HD paper \cite{HDC}, and our vertex covers (Definition~\ref{d:vc}) are all the same, but differ in important details. Despite finding the SPHS more cumbersome than our vertex covers, due to the use of $\mathcal{P}(r)$ which has a relatively complex definition, we include the definition so that we can show how our vertex covers combined with long edges give a SPHS (Lemma~\ref{l:vclesphs}), and thus immediately apply those lemmas of \cite{HD} that apply to SPHSs to the output of our algorithm.

\begin{defn}\jlabel{d:sphs}
{\bf Sparse shortest path hitting set (SPHS)}
An $(r,h)$-SPHS is a set of vertices $C \subseteq V$ such that \begin{enumerate}
	\item for all $p$ in $\mathcal{P}(r)$, $C\cap V(p) \not = \emptyset$ and
	\item for all $v '\in V$, $|C \cap B(v,2r)|\leq h$.
\end{enumerate}
\end{defn}

\pagebreakk
\section{Lemmas about the definitions}

In this section we establish a number of facts about the definitions of the previous section, and in particular for graphs of constant highway dimension. 

\subsection{Facts about highway dimension}

\jlabel{s:hdf}

Here we show that graphs of highway dimension have constant degree, and thus $|E|=\Theta(|V|)$ (recall that we only consider connected graphs), and then establish the relationships between the various variants of highway dimension.

\begin{lemma}\jlabel{l:hddegree}{\bf Degree bound, Lemma 3.5 of \cite{HD}.}
A graph of continuous highway dimension ${h}$ has maximum degree ${h}$.
\end{lemma}


\begin{figure}
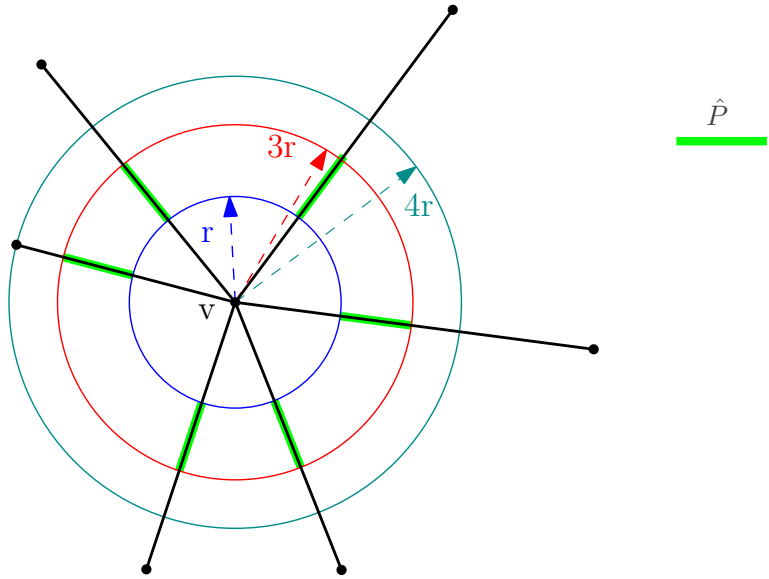

    \centering
    \finalfigfrompage{3}{0}
    \caption{Illustration of the proof of Lemma~\ref{l:hddegree}.}
    \label{fig:3}
\end{figure}

\begin{proof}
See Figure~\ref{fig:3}.
Consider a vertex $v$ with degree $d$ and shortest incident edge of size $4r$. Define a set of $d$ paths $\hat{P}$ on each incident edge of $v$ from distance $r$ to distance $3r$ from $v$. As these paths are disjoint, of length greater than $r$, and inside $B(v,4r)$, the highway dimension must be at least $h$.
\end{proof}

\begin{lemma} \jlabel{l:dhdhd}
{\bf Continuous highway dimension implies discrete highway dimension.}
A graph with continuous highway dimension $h$ has discrete highway dimension at most $h$.
\end{lemma}

\begin{figure}
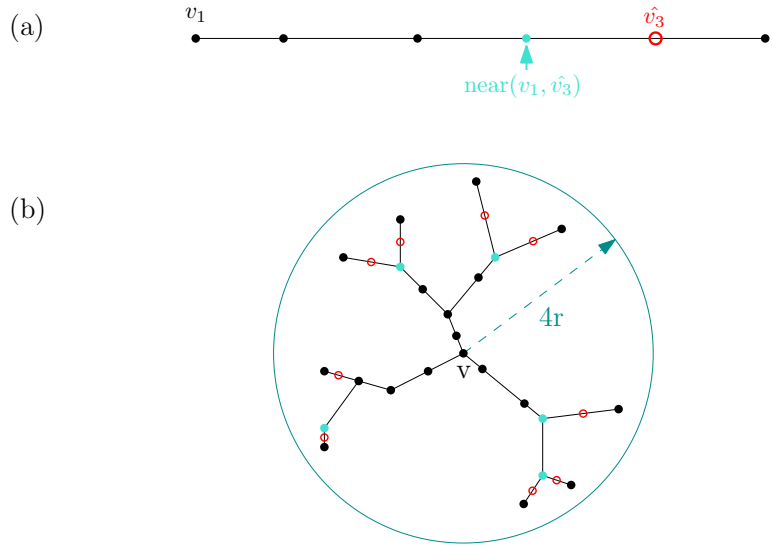

    \centering
    \finalfigfrompage{4}{0}
    \caption{Illustration of the proof of Lemma~\ref{l:dhdhd}. In (b), the cyan vertices are the set $C$ and the red circles indicate elements of $\hat{C}$, which may lie on edges.}
    \label{fig:4}
\end{figure}

\begin{proof}
Fix a $v \in V$ and $r$. See Figure~\ref{fig:4}. Given a vertex $\hat{v_3} \in \hat{B}(v,4r)$, define $\near(v,\hat{v_3})$ to be the last vertex in $V$ on the shortest path from $v$ to $\hat{v_3}$; if $\hat{v_3} \in V$, then $\near(v,\hat{v_3})=\hat{v_3}$ and in general $d(v,\near(v_3)) \leq {d}(v,\hat{v}_3)$.

Let $\hat{C}$ be the set of size at most $h$ that must exist by continuous highway dimension, and let $C \coloneqq\bigcup_{\hat{c} \in \hat{C}}\near(v,\hat{c})$. Note that since $\hat{C} \subseteq \hat{B}(v,4r)$, $C \subseteq B(v,4r)$. Also note $|C|\leq |\hat{C}|\leq h$.

Consider any $v_1,v_2 \in V$ such that $V(P(v_1,v_2)) \subseteq B(v,4r)$. By continuous highway dimension, there is some $\hat{c}\in\hat{C}$ on $\hat{p}(v_1,v_2)$. Since if $\hat{c} \not = \near(v,\hat{c})$ then $\near(v,\hat{c})$ is simply one of the two vertices of the edge $\hat{c}$ lies on, $\near(v,\hat{c})$ is also an element of $V(p(v_1,v_2))$ and $C$.
 \end{proof}

\begin{lemma}\jlabel{l:hdchd}
{\bf Highway dimension and continuous highway dimension are factor-2 equivalent (Theorem 11.3 of \cite{HD})}
If $G$ has highway dimension $h$ and continuous highway dimension $h'$, then $h \leq h' \leq 2h$.
\end{lemma}

We omit the proof and refer the interested reader to Theorem 11.3 of \cite{HD}.

\begin{cor} \jlabel{cor:hd} {\bf Constant highway dimension implies constant discrete highway dimension and continuous highway dimension.}
	Note that as we have assumed graph $G$ has highway dimension $O(1)$, by Lemmas~\ref{l:dhdhd} and~\ref{l:hdchd} we may assume that $G$ has discrete highway dimension and continuous highway dimension $O(1)$.
\end{cor}

\pagebreakk
\subsection{Facts about a vertex cover}

Here we show that graphs of constant highway dimension have a vertex cover. This is similar in spirit to the proofs found in the HD papers, but the proof is completely different due the the subtleties of our definition of vertex cover.

\begin{lemma} \jlabel{l:hdpvc}
{\bf Discrete highway dimension implies vertex cover.}
Every graph of discrete highway dimension $h$ has a $(r,h)$  vertex cover for every $r$.
\end{lemma}

\begin{figure}
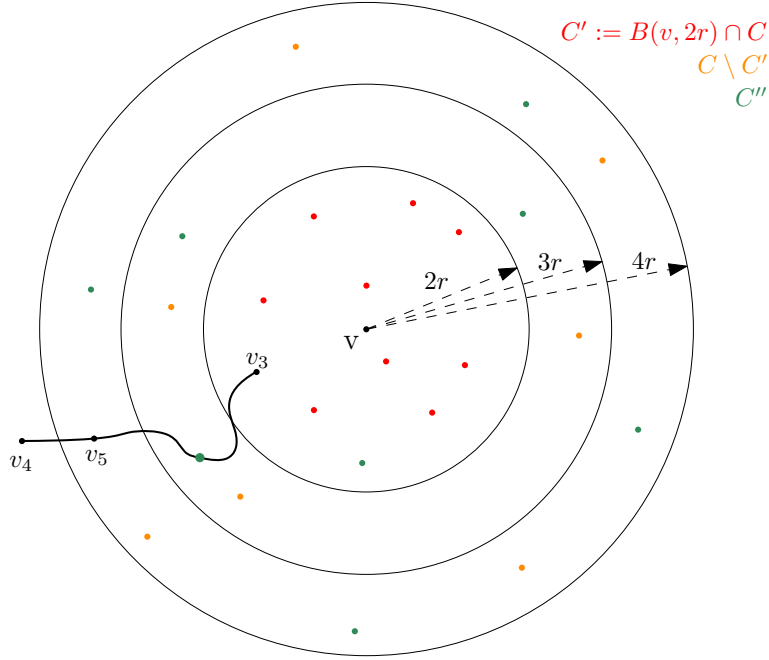

    \centering
    \finalfigfrompage{5}{0}
    \caption{Illustration of point (\ref{p3}) in the proof of Lemma~\ref{l:hdpvc}.}
    \label{fig:5}
\end{figure}

\begin{proof}
Let $C$ be a minimum cardinality subset of $V$ that satisfies (1) of the definition of vertex cover (Definition~\ref{d:vc}): 
all shortest paths $p(v_1,v_2)$, $v_1,v_2 \in V$, of length $d(v_1,v_2) > r $ and $\maxedge(p(v_1,v_2))
\leq r$ have  some $v_c\in C$ where $v_c \in V(p(v_1,v_2))$. Such a set always exists as the set $V$ of all vertices trivially satisfies the covering condition; here we consider the minimal one.

If $C$ satisfies (2) of the definition of $(r,h)$ vertex cover, which requires that for all $v \in V$ $|B(v,2r) \cap C|\leq h$, then the lemma is proven. We will show that if $C$ does not satisfy (2), then $C$ is not  a minimum cardinality subset of $V$ that satisfies (1) of the definition of vertex cover and thus will obtain a contradiction.

As (2) of Definition~\ref{d:vc}  is not satisfied, there must be a vertex $v\in V$ such that $|B(v,2r) \cap C|>h$, let $C' \coloneqq B(v,2r) \cap C$. 

Let $C''$ be a set of at most $h$ vertices in $B(v,4r)$ that intersect all length at least $r$ shortest paths between vertices in $V$ that lie entirely inside $B(v,4r)$; this must exist by the definition of discrete highway dimension (Definition~\ref{d:dhd}). 

Observe that since $|C'| >  h$ and $|C''|\leq h$, and thus $|C \setminus C' \cup C''| < |C|$, if we can show $C \setminus C' \cup C''$ satisfies (1) of the definition of vertex cover (Definition~\ref{d:vc}), that would contradict the minimality of $C$ and prove the lemma.

We will now argue that set $C \setminus C' \cup C''$ will thus have a nonempty intersection with all shortest paths $p(v_1,v_2)$ where $d(v_1,v_2) \geq r$ and $\maxedge(p(v_1,v_2)) \leq r$.
This has several cases:

\begin{enumerate}
    
\item 
First, if $V(p(v_1,v_2)) \subseteq B(v,4r)$, since $d(v_1,v_2)\geq r$, then $C''$ will intersect $V(p(v_1,v_2))$, and thus so will $C \setminus C' \cup C''$. 

\item 
Second if $V(p(v_1,v_2)) \cap  B(v,2r) = \emptyset$, then we know there is some element $c\in C$ that intersects $V(p(v_1,v_2))$ and that $c$ is not an element of $C'$ since $C' \subseteq B(v,2r)$; thus $c$ remains in $C \setminus C' \cup C''$.

\item \label{p3}
See Figure~\ref{fig:5}. Thus, thirdly, $p(v_1,v_2)$ must contain some vertex inside $B(v,2r)$, call it $v_3$, and some vertex outside $B(v,4r)$, call it $v_4$. If there are multiple choices for $v_3$ and $v_4$, choose them to minimize $d(v_3,v_4)$; this ensures that $V(p(v_3,v_4)) \setminus \{v_3,v_4\}$ contains neither vertices in $B(v,2r)$ nor vertices not in $B(v,4r)$. 

Let $v_5$ be the vertex adjacent to $v_4$ on $p(v_4,v_3)$. Vertex $v_5$ is in the annulus $B(v,4r) \setminus B(v,3r)$; this is because $v_4$ is outside $B(v,4r)$, and $v_4$ and $v_5$ are connected by an edge with $\maxedge(v_4,v_5)\leq r$.

Note that as $v_3$ is in $B(v,2r)$ thus $d(v_3,v_5) \geq r$. Also, by construction $V(p(v_3,v_5)) \subseteq B(v,4r)$. Thus by the same logic as in the first case, $C''$ will intersect $V(p(v_3,v_5))$ which is a subset of $V(p(v_1,v_2))$, and thus so will $C \setminus C' \cup C''$.
\end{enumerate}
\end{proof}



\subsection{Facts related to doubling dimension}

One of the greatest challenges is how to deal with edges of varying lengths, and this lemma shows that there are not many long edges that intersect any ball.

\begin{lemma} {\bf Big edges in a ball.}
\jlabel{lem:bigedge}
In a graph of constant highway dimension, for any $r>0$,
there are at most $O(1)$ edges of length at least $r$  with at least one vertex in $B(v,2r)$. 
\end{lemma}

\begin{proof}
Fix $v$ and $r$ and
let $D$ denote the set of edges of length at least $r$ with at least one vertex in $B(v,2r)$.
As we require edges to be useful, each edge of length at least $r$ in $D$ is also a shortest path of length at least $r$ and is in $S_r$.  Thus, by the definition of discrete highway dimension, there is a set $C$ of size at most $h$ such that each edge in $D$ has as an endpoint at least one vertex in $C$. As the maximum degree is $h$ by Lemma~\ref{l:hddegree}, each of the at most $h$ elements of $C$ can only be adjacent to $h$ edges in $D$; thus $|D| \leq h^2=O(1)$.
\end{proof}



\newcommand{\edgeconst}{c_{\rm edge}}
\newcommand{\edgeconstval}{2h\constdd(h)^5}

We now show the proof that any graph of constant highway dimension has constant doubling dimension. However, the converse is easily seen to not be true in the case of a uniform grid; however large uniform grids are not seen in real word road network, with an emphasis on the word uniform.

\begin{lemma} {\bf Doubling dimension, Theorem~9.1 of~\cite{HD}.}\jlabel{l:doubling}
Given a graph of constant highway dimension, then any ball $B(v,2r)$ can be covered with $O(1)$ balls of size $r$. 
Precisely, for every $v$ there exists a set $V' \subseteq V$, $|V'| =O(1)$ such that $B(v,2r) \subseteq \cup_{v'\in V'} B(v',r)$.
\end{lemma}

\begin{proof}
Let $h$ be the discrete highway dimension of $G$, this is $O(1)$ by Corollary~\ref{cor:hd}.
Let $C$ be vertices in a $(r,h)$ vertex cover of $G$ that intersect $B(v,2r)$, such a vertex cover exists by Lemma~\ref{l:hdpvc}. 
Let $D$ be the vertices of edges entirely in $B(v,2r)$ that have length at least $r$, there are at most $O(1)$ such edges Lemma~\ref{lem:bigedge}.
We argue that the balls of size $r$ centered at the elements of $V' \coloneqq D \cup C \cup \{v\}$ cover $B(v,2r)$. As $B(v,r)$ is covered from $v$, we need only worry about $B(v,2r) \setminus B(v,r)$. Consider some $v'$ in $B(v,2r) \setminus B(v,r)$. The path $p(v',v)$ is in the ball $B(v,2r)$ and has length at least $r$. In the first $r$ of $p(v',v)$ there is either an edge of length $>r$, in which case there is a vertex of $D$, or not, in which case since the maximum edge is less than $r$ there must be an element of $c\in C$ of $C$ in the first $r$ of $p(v',v)$. In either case $v' \in B(v,r)$.
\end{proof}

We can recurse on this lemma so that it can be applied to ball size ratios other than two:

\begin{cor} \jlabel{cor:double}
{\bf Generalized doubling. }
In a graph with constant highway dimension, for any $v\in V$ and $r\geq r'>0$ such that $\frac{r}{r'}=O(1)$, there is a set $V' \subseteq V$ of size $O(1)$ such that $B(v,r) \subseteq \bigcup_{v' \in V'}B(a,v')$. 
\end{cor}

%

%

\subsection{Facts about shortcut graphs}

The shortcut graph has been constructed so that in can be used instead of the original graph between edges of $C$ that are at most distance $r$ apart.

\begin{obs} \jlabel{o:shortcutfacts} 
Let $C$ be a subset of $V$. For any $v_1,v_2 \in C$, and positive $r$, if $\maxedge(p(v_1,v_2)) \leq r$:
\begin{itemize}
    \item $d_G(v_1,v_2)=d_{G(C,r)}(v_1,v_2)$
    \item If vertex $v_3 \in C$ is on $p_{G(C,r)}(v_1,v_2)$, then $v_3$ is also on $p_{G}(v_1,v_2)$
\end{itemize}
\end{obs}

Now, we show that vertices of shortcut graphs of vertex covers have small degree:

\begin{lemma}{\bf Degree of shortcut graph, similar to Lemma 4.2 of~\cite{HDC}}\jlabel{l:degree}.
Given a $(r,k)$ vertex cover $C$ of $G$, and a positive $r'$ such that $\frac{r'}{r}=O(1)$, the shortcut graph $G(C,r')$ has maximum degree $O(1)$.    
\end{lemma} 

\begin{proof}
In $G(C,r')$, each vertex $v \in C$ is only connected to vertices $v'\in C$ which are in $B(v,r')$. By the definition of vertex cover and Corollary~\ref{cor:double} there are only $O(1)$ such vertices.
\end{proof}

Additionally, we show that the complexity of all items in a ball is small; this is important as our data structures in next section will have a size proportional to degrees and vertices of shortcut graphs :

\begin{lemma}
	\jlabel{l:shorcutsize}
 {\bf Shortcut size.}
	Given a $(r,k)$ vertex cover $C$ of $G$ and a positive $r'$ such that $\frac{r'}{r}=O(1)$, for any $v$ the complexity (the sum of vertex degrees) of $B(v,r') \cap C$ in $G(C,r')$ is $O(1)$.
\end{lemma}

\begin{proof}
	As noted in the proof of Lemma~\ref{l:degree}, there are only $O(1)$ vertices in $B(v,r') \cap C$, and by Lemma~\ref{l:degree} they have degree $O(1)$.
\end{proof}

\pagebreakk

\subsection{Facts about a path cover}


Path covers are a completely new idea of this work.

\begin{lemma} \jlabel{l:hdpc}
{\bf Highway dimension implies path cover.}
A graph with constant highway dimension has a $(r,O(1))$ path cover for any $r>0$ (the constant in the $O(1)$ depends on the highway dimension but it is independent of $r$).
\end{lemma}

\begin{proof}
Let $h$ be the discrete highway dimension of $G$; $h=O(1)$ by Corollary~\ref{cor:hd}.
Let $C$ be a $(\frac{1}{8}r,h)$ vertex cover, which exists by Lemma~\ref{l:hdpvc}. 
Let $P$ be the set of the shortest paths between elements of $C$ that are of length between $\frac{1}{4} r$ and $r$. Formally:
$$ P \coloneqq \left\{ p(x_1,x_2) \middle| 
v_1,v_2 \in C  \text{ and } \frac{1}{4} r \leq  d(v_1,v_2)\leq  r \right\} $$

We argue that $P$ is a $r,O(1)$ path cover. This requires arguing each requirement to be a path cover holds:

\begin{enumerate}
        \item \emph{$P$ is a set of shortest paths between vertices in $V$, each of length from $\frac{1}{4} r$ to $r$.}

This is by construction.


\begin{figure}
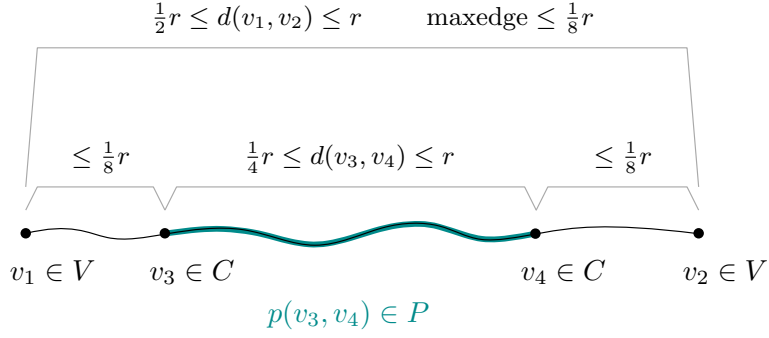

    \centering
    \finalfigfrompage{6}{0}
    \caption{Illustration of point (\ref{pp2}) of Lemma~\ref{l:hdpc}.}
    \label{fig:6}
\end{figure}

\item \label{pp2} \emph{For any $v_1,v_2\in V$ where the distance $\frac{1}{2}r \leq d(v_1,v_2)\leq r$ and $\maxedge(p(v_1,v_2)) \leq \frac{1}{8}r$, $p(v_1,v_2)$ contains as a subpath some path $p(v_3,v_4) \in P$, where 
$d(v_1,v_3)\leq \frac{1}{8}r$ and
$d(v_4,v_2)\leq \frac{1}{8}r$.}

See Figure~\ref{fig:6}. Let $p$ be some path $p(v_1,v_2)$ where the distance $\frac{1}{2}r \leq d(v_1,v_2)\leq r$ and $\maxedge(p(v_1,v_2)) \leq \frac{1}{8}r$. Let $v_3$ be a vertex of $C$ on $p$ such that $d(v_1,v_3)\in [0,\frac{1}{8}r]$ and let $v_4$ be a vertex of $C$ on $p$ such that $d(v_4,v_2)\in [0,\frac{1}{8}r]$.

These exist by the definition of a vertex cover, are distinct, and meet the distance requirements. By construction, the distance between $v_3$ and $v_4$ is between $\frac{1}{4}r$ and $r$ and thus $p(v_3,v_4)$ is an element of $P$.

    \item \emph{Any ball of radius $4 r$ contains at least one vertex of a constant number paths in $P$.}

By Corollary~\ref{cor:double}, each ball $B(v,4 r)$ can be covered by 
$O(1)$ balls of radius $\frac{1}{8}r$. Each of these balls intersects at most $h$ elements of $C$, which bounds $|B(v,5 r)\cap C|$ by $O(1)$. A path in $P$ with any vertex in $B(v,4 r)$ must have both endpoints in $B(v,5 r)$, and thus there are at most $|B(v,5 r)\cap C|^2 =O(1)$ such paths.
\end{enumerate}

\end{proof}

\subsection{Facts about a SPHS}
Here we link our vertex covers and show that together with long edges one obtains a SPHS of \cite{HD}.

\begin{lemma}
\jlabel{l:vclesphs}
{\bf A vertex cover and long edges give a SPHS}
	Let $E'$ be the set of edges in $G$ of length at least $r/3$.
	Let $C$ be a $(r/3,k)$ vertex cover. 
	Then $C \cup V(E')$ is a $(r,O(1))$-SPHS, given $G$ has constant highway dimension.
\end{lemma}

\begin{proof}
We argue each requirement of a SPHS separately:
\begin{enumerate}
	\item \emph{For all $p$ in $\mathcal{P}(r)$, $C\cap V(p) \not = \emptyset$.}
	
		Consider some path $p$ in $\mathcal{P}_r$ and let $p'$ be its $r$-witness. If $\maxedge(p') \geq \frac{r}{3}$, then $p$ will contain a vertex of $E'$. Otherwise, we know that the length of $p'$ is at least $r$, and the up to two edges in $p'$ but not in $p$ are at most $\frac{r}{3}$, thus $p$ is has length least $\frac{r}{3}$; since $\maxedge(p') \leq \frac{r}{3}$, there must be an element of $C$ on $p$.

	\item \emph{For all $v '\in V$, $|C \cap B(v,2r)| = O(1)$.}

	The number of elements in $|B(v,r/3)\cap C|$ is at most $k$ by Definition~\ref{d:vc}. By  Corollary~\ref{cor:double}, there are at most $O(1)$ elements in $|B(v,r)\cap C|$.
	
	The number of elements in $|E' \cap B(v,2r)|$ has at most $h^2$ edges inside by Lemma~\ref{lem:bigedge}. Applying Corollary~\ref{cor:double}, $|E' \cap B(v,4r)|$ is at most $O(1)$.
	
\end{enumerate}

\end{proof}

\pagebreakk\section{Creating a SPHS hierarchy}

In this section we describe how to create a SPHS hierarchy, which will use a hierarchy of vertex covers, endpoints of long edges, and their associated shortcut graphs. We focus on a combinatorial description here, and will address algorithmic issues in Section~\ref{s:computing}.

\subsection{Definitions}

We describe the edge sets $E[i]$, the graphs $G[i]$, and the vertex sets $C[i]$ and $C'[i]$, where $i$ an integer at least $-1$. 

$E[i]$ is the set of all edges of length in $(8^{i-1},8^{i}]$; recall that all edges have at least unit length, thus the $E[i]$ partition the edges and $E[-1]=\emptyset$. We use $E[\geq i]$ to denote $\bigcup_{j=i}^\infty E[j]$.

The graph $G[-1]$ is empty as is the vertex set $C'[-1]$. The set $C[i]$ is defined to be $C[i] \coloneqq C'[i] \cup V(E[\geq i])$, with the goal that 
 $C[i]$ is a $(8^i,O(1))$ vertex cover assuming constant highway dimension $h$, which we will prove
in Lemma~\ref{l:ciscover}. $G[i]$ is the shortcut graph $G(C[i],8^i)$.

The set $C'[i]$ is defined based on $C[i-1]$; note as there is some flexibility as to what can be in $C'[i]$ we speak of a \emph{valid} $C'[i]$ as one that meets the requirements below:

\begin{defn}
\jlabel{d:construct}
$C'[-1] \coloneqq \emptyset$. For integer $i\geq 0$
a valid $C'[i]$ is any subset of $C[i-1]$ that can be produced by the following method:
\begin{itemize}
    \item Initialize $C'[i]$ to be $\emptyset$
    \item For each $v,v' \in V(G[i-1])$  such that $d_{G[i-1]}(v,v')\in [\frac{3}{4}\cdot8^i,8^{i}]$ and $\maxedge_G(p(v,v'))\leq 8^{i-1}$, if the path $p_{G[i-1]}(v,v')$ does not yet contain any element of $C'[i]$, add the vertex of $G[i-1]$ closest to the midpoint of $p_{G[i-1]}(v,v')$ to $C'[i]$. This can be done using $G[i-1]$ by Observation~\ref{o:shortcutfacts}.
\end{itemize}
Note that the order in which all possible $(v,v')$ pairs are considered is arbitrary, and in general different orderings will yield different $C'[i]$'s, all of which are valid.
\end{defn}

\pagebreakk

\subsection[{The C[i]'s are vertex covers}]{The $C[i]$'s are vertex covers}

\begin{lemma} {\bf{$C[i]$ is a vertex cover.}}
\jlabel{l:ciscover}
    Each $C[i]$ is a $(8^i,O(1))$ vertex cover of $G$.
\end{lemma} 

\begin{proof}
We proceed by induction on $i$. First, the base case of $C[-1]$, where we argue each requirement of a vertex cover separately:
\begin{enumerate}
    \item \emph{All shortest paths $p(v_1,v_2)$, $v_1,v_2 \in V$, of length $d(v_1,v_2) > 8^{-1} $ and $\maxedge(p(v_1,v_2))
\leq 8^{-1}$ have  some $v_c\in C$ where $v_c \in V(p(v_1,v_2))$.}

This holds trivially as we assume all edges have at least unit length.

\item \emph{For all $v \in V$ $|B(v,2) \cap C[-1]|=O(1)$.}

As $C'[-1]$ is defined to be $\emptyset$ and $E[-1]$ is also $\emptyset$ due to the minimum unit edge length requirement, this holds trivially as $C[-1]=C'[-1] \cup E[-1] = \emptyset
$.

\end{enumerate}

We now prove $C[i]$ is a $(8^i,O(1))$ vertex cover of $G$, given $C[i-1]$ is a $(8^{i-1},O(1))$ vertex cover of $G$ (we note the $O(1)$ is the same for all values of $i$ and depends only on the highway dimension). We argue each requirement of a path cover separately.

\begin{enumerate}
\item  \label{ppp1}
\emph{All shortest paths $p(v_1,v_2)$, $v_1,v_2 \in V$, of length $d(v_1,v_2) > 8^i $ and $\maxedge(p(v_1,v_2))
\leq 8^i$ have  some $v_c\in C$ where $v_c \in V(p(v_1,v_2))$.}

Fix some path $p(v_1,v_2)$ of length $d(v_1,v_2) > 8^i $ where $\maxedge(p(v_1,v_2)) \leq 8^i$. 

There are two cases, depending on the relationship of $\maxedge(p(v_1,v_2))$ with $8^{i-1}$. If $\maxedge(p(v_1,v_2)) > 8^{i-1}$, then 
$ 8^{i-1} < \maxedge(p(v_1,v_2)) \leq   8^{i} $. Thus there is some edge in $E[i]$ on $p(v_1,v_2)$, and since $C[i] \coloneqq C'[i] \cup V(E[i])$ this completes the proof of this case.

\begin{figure}
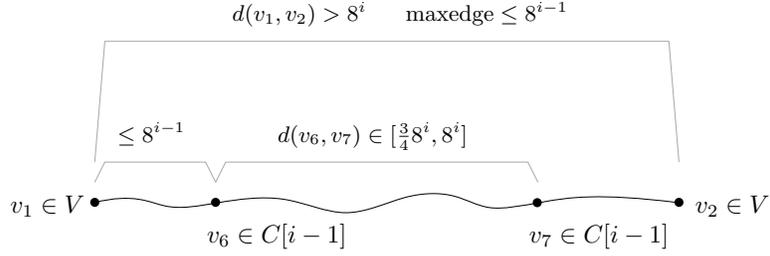

    \centering
    \finalfigfrompage{7}{0}
    \caption{Case where $\maxedge(p_G(v_1,v_2)) \leq 8^{i-1}$ in (\ref{ppp1}) of Lemma~\ref{l:ciscover}.}
    \label{fig:7}.
\end{figure}

Otherwise, $\maxedge(p_G(v_1,v_2)) \leq 8^{i-1}$. See Figure~\ref{fig:7}. Thus, as $C[i-1]$ is an $(8^{i-1},O(1))$ vertex cover, there must be vertices $v_6,v_7 \in C[i-1]$ on $p_G(v_1,v_2)$ such that $d_G(v_1,v_6) \leq 8^{i-1}$ and $d_G(v_1,v_7) \in [7\cdot8^{i-1},8\cdot 8^{i-1}]$.
This implies that $d_G(v_6,v_7) \in [\frac{3}{4}8^i, 8^i]$; thus the path $p_{G[i-1]}(v_6,v_7)$ will be considered by the algorithm that constructs $C'[i]$ and a vertex on the path will be added if it was not there already.

\item \label{ppp2}
\emph{For all $v \in V$ $|B(v,2\cdot 8^i) \cap C[i]|=O(1)$.}

We first bound $|B(v,2\cdot 8^i) \cap C'[i]|$.

    Let $P$ be a $(8^i,O(1))$ path cover of $G$, which exists by Lemma~\ref{l:hdpc}.
    Consider a function $\rho$ that maps every path $p(v_1,v_2)$, $d(v_1,v_2)\in [\frac{3}{4}8^i,8^i]$ to some $p(v_3,v_4) \in P$, $v_3,v_4 \in V$, that is a subpath of $p(v_1,v_2)$ and where $d(v_1,v_3)\leq \frac{1}{8}8^i$ and $d(v_4,v_2)\leq \frac{1}{8}8^i$; this exists by Lemma~\ref{l:hdpc} and Definition~\ref{d:pc}.


\begin{figure}
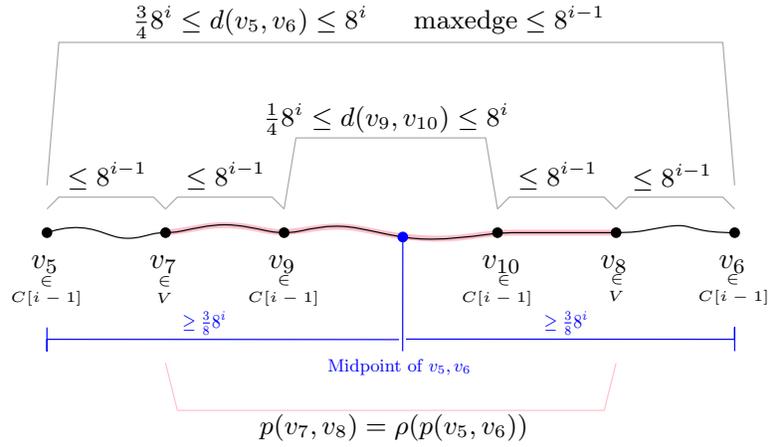

    \centering
    \finalfigfrompage{8}{0}
    \caption{Illustration of when the algorithm for constructing $C'[i]$ adds an element $c$ to $C'[i]$ inside $B(v,2\cdot 8^i)$ after looking at some $v_5,v_6 \in C[i-1]$ and finding no elements of $C'[i]$ on $p(v_5,v_6)$ in (\ref{ppp2}) of the proof of Lemma~\ref{l:ciscover}.}
    \label{fig:8}
\end{figure}

    Consider when the algorithm for constructing $C'[i]$ adds an element $c$ to $C'[i]$ inside $B(v,2\cdot 8^i)$ after looking at some $v_5,v_6 \in C[i-1]$ and finding no elements of $C'[i]$ on $p(v_5,v_6)$. See Figure~\ref{fig:8}. Recall the algorithm only considers $p(v_5,v_6)$ if $\maxedge(v_5,v_6)\leq 8^{i-1}$. 
    We know $d(v_5,v_6)\in [\frac{3}{4}8^i,8^i]$ by the definition of the algorithm. Thus path $\rho(p(v_5,v_6))$ is well defined, which we will denote as $p(v_7,v_8)$, and $d(v_5,v_7)\leq \frac{1}{8}8^i$ and $d(v_6,v_8)\leq \frac{1}{8}8^i$. Let $v_9$ and $v_{10}$ be vertices of $C[i-1]$ on $p(v_7,v_8)$ such that $d(v_7,v_9) \leq \frac{1}{8}8^i$ and $d(v_8,v_{10}) \leq \frac{1}{8}8^i$; these exist as by assumption $C[i-1]$ is an $(8^{i-1},O(1))$ vertex cover. Thus since $d(v_5,v_9) \leq \frac{1}{4}8^i$, and $d(v_5,v_6)\in [\frac{3}{4}8^i,8^i]$, the vertex $v_9$ is on the first half of $p(v_5,v_6)$; symmetrically, $v_{10}$ is on the second half of $p(v_5,v_6)$. The algorithm will choose $c$ such that it is the element of $C[i-1]$ on $p(v_5,v_6)$ closest to the midpoint of $p(v_5,v_6)$, since $v_{9}$ and $v_{10}$ are elements of $C[i-1]$ and are on $p(v_7,v_8)$ which is a subpath of $p(v_5,v_6)$, we know $c$ is on 
    $p(v_7,v_8) = \rho(p(v_5,v_6))$ and $\rho(p(v_5,v_6))$ previously had no elements of $C'[i]$. Furthermore $\rho(p(v_5,v_6))$ is inside $B(v,3\cdot 8^i)$ as $c$ is inside $B(v,2\cdot 8^i)$ and $c$ is on $\rho(p(v_5,v_6))$ which has maximum length $8^i$.

Thus the size of $C'[i] \cap B(v,2 \cdot 8^i)$ is at most the number of paths in $P$ that intersect $B(v,3 \cdot 8^i)$. As $P$ is a $(8^i,O(1))$ path cover, we know that there are at most $O(1)$ paths in $P$ that intersect any ball $B(v,4 \cdot 8^i)$.

We have thus bounded the size of $C'[i] \cap B(v,2 \cdot 8^i)$ to be $O(1)$, and Lemma~\ref{lem:bigedge} and Corollary~\ref{cor:double} bound the size of $E[i] \cap B(v,2 \cdot 8^i)$ to be $O(1)$.

Recalling that $C[i]=C'[i] \cup E[i]$ completes the lemma.
\end{enumerate}

\end{proof}
%
%
%
%

\subsection[{The C[i]'s are SPHS's}]{The $C[i]$'s are SPHS's}

\begin{lemma}
\jlabel{l:cisphs}
	The $C[i]$'s are SPHS's.
\end{lemma}

\begin{proof}
We know that $C[i]$ is a $(8^i,O(1))$ vertex cover.
We also know that $C[i]$ contains the endpoints of all edges of length at least $8^{i-1}$. Thus by Lemma~\ref{l:vclesphs}, $C[i]$ is a $(3\cdot 8^i,O(1))$-SPHS.
\end{proof}

\pagebreakk

\section{Bounding the size and height of the structure}

Given $C \subseteq V$ and $r>0$, let  $\ballcover_G(C,r)$ be a minimum cardinality subset $C'$ of $C$ such that 
$$ C \subseteq \bigcup_{v\in C'}B_G(v,r).$$

Let $BC[i] \coloneqq \ballcover(C[i],8^i)$, and let $IG[i]$ be the graph with $BC[i]$ as vertices and edges between elements of $BC[i]$ if their distance is at most $3 \cdot 8^i$:

    $$IG[i]\coloneqq (BC[i],\{\{v_1,v_2\}|v_1,v_2 \in  BC[i]\text{ and } d(v_1,v_2) \leq 3 \cdot 8^i ).$$

\begin{lemma}\jlabel{l:connected}{\bf High levels of $IG[i]$ are connected.}
 For all $i$ such that $\log_8 U < i $ the graph $IG[i]$ is connected.
\end{lemma}

\begin{proof}
Fix $i$.
Let the witness ball of $v$, $\beta(v)$, be some element $v'$ of $BC[i]$ such that $v \in B(v',8^i)$, with the specific condition that if $v\in BC[i]$, $\beta(v)=v$.
Consider any two $v_1,v_2$ that are adjacent in $G$. We argue that $\beta(v_1)$ and $\beta(v_2)$ are adjacent in $IG[i]$. Note that $d(\beta(v_1),v_1)$ and $d(\beta(v_2),v_2)$ are both at most $8^i$ by construction. We also know that $d(v_1,v_2)\leq U \leq 8^k$, and thus $d(\beta(v_1),\beta(v_2)) \leq 3 \cdot 8^i$, which means that $\beta(v_1)$ and $\beta(v_2)$ are connected in $IG[i]$.

This shows that any $v_3,v_4 \in BC[i]$ are connected in $IG[i]$ as for each adjacent $v_1,v_2$ on $p(v_3,v_4)$, their witness balls $\beta(v_1)$ and $\beta(v_2)$ are connected in $IG[i]$ and $\beta(v_3)=v_3$ and $\beta(v_4)=v_4$.
\end{proof}

\begin{lemma} \jlabel{l:mdig} {\bf Maximum degree of intersection graph.}
The maximum degree of $IG[i]$ is $O(1)$.
\end{lemma}

\begin{proof}
If $v_1,v_2$ is an edge in $IG[i]$, thus $v_1,v_2 \in C[i]$ and $d(v_1,v_2) \leq 3 \cdot 8^i$.
By Corollary~\ref{cor:double}, and Lemma~\ref{l:ciscover} which shows that $C[i]$ is a $(8^i,O(1))$ vertex cover, there are only $O(1)$ elements of $C[i] \cap B(v_1,3 \cdot 8^i)$.
\end{proof}

\begin{lemma} \jlabel{l:tur} {\bf Relating covers of different size balls.}
There is a positive constant $c_{vc}<1$ such that for all $i> \log_8 U$, there is a set $VC[i]$ such that
$C[i] \subseteq \bigcup_{v \in VC[i]}B(v,8^{i+1})$ and 
$|VC[i]| \leq c_{vc}|BC[i]|$. That is, $\ballcover(C[i],8^{i+1}) \leq 
c_{vc} \cdot |\ballcover(C[i],8^{i})|$.
\end{lemma}

\begin{proof}
    Turan's theorem states that any graph $G$ of average degree $d_G$ has an independent set of size at least $\frac{1}{1+d_G}|V|$ and a vertex cover of size at most $\frac{d_G}{1+d_G}|V|$. Thus, as $IG[i]$ has average degree at most $O(1)$, $IG[i]$ has a vertex cover of size at most $c_{vc}|BC[i]|$ for some positive $c_{vc}<1$. Call such a cover $VC[i]$. 

    Observe that 
  $C[i] \subseteq \bigcup_{v \in VC[i]} B(v,3\cdot 8^i)$ since every element of $C[i]$ is either in $B(v,8^i)$ for some $v \in VC[i]$ or in some $B(v',8^i)$ that intersects $B(v,8^i)$. This requires that $VC[i]$ is connected as proved in Lemma~\ref{l:connected}, as a vertex cover simply ensures that there is one vertex in the cover adjacent to every edge and does not require that degree-0 vertices be in the cover.
\end{proof}




\begin{lemma}
\jlabel{l:bcs} {\bf Ball cover sizes are geometric.}
When $i\geq \log_8 U +1$,
    $|BC[i]| \leq c_{vc} \cdot BC[i-1]$
\end{lemma}

\begin{proof}
We use the fact that by construction $C[i] \subseteq C[i-1]$:
\begin{align*}
    |BC[i]| &= |\ballcover(C[i],8^i)|
    & \text{Definition}
\\
&\leq \ballcover(C[i-1],8^i) 
& C[i] \subseteq C[i-1]
\\
&\leq c_{vc} \ballcover(C[i-1],8^{i-1})  
& \text{Lemma~\ref{l:tur}}
\end{align*}
\end{proof}

\begin{lemma}\jlabel{l:size} {\bf Size of the structure}
The sum of the sizes of the nonempty elements of $C[i]$, $C'[i]$, $E[i]$, and $G[i]$ are all $O\left( |V| \log U \right)$ for constant $h$.    
\end{lemma} 
\begin{proof}
It suffices to bound $C[i]$ as $C'[i]$ is a subset of $C[i]$, and $G[i]$ has elements of $C[i]$ as vertices and its vertices have degree at most $O(1)$ by Lemma~\ref{l:degree}. The $E[i]$ partition the edges, and $|E|$ is at most $O(|V|)$ by Lemma~\ref{l:hddegree}.

Let $k=\lceil \log_8 U \rceil +1 $. Every $C[i]$ has size at most $|V|$, thus $\sum_{i=0}^{k-1} |C[i]|=O(|V| \log U)$.


We know $|BC[k]|\leq |V|$ and when $i>k$, there is a positive $c_{vc}<1$ such that $|BC[i]| \leq c_{vc} \cdot |BC[i-1]|$ by Lemma~\ref{l:bcs}. Thus $\sum_{i=k}^\infty |BC[i-1]| = O(|V|)$.
Additionally, $ |BC[i]| =\Theta(C[i])$, since each ball in $BC[i]$ must intersect at least 1 and at most a constant number of elements of $C[i]$.  Thus   $\sum_{i=k}^\infty C[i] =O(|V|)$.
\end{proof}

\begin{lemma}\jlabel{l:height}{\bf Height of the structure.}
The maximum nonempty $C[i]$ is $h=O(\log (|V|+U))$.
\end{lemma}

\begin{proof}
Let $k=\lceil \log_8 U \rceil +1 $.
    We know $|BC[k]|\leq |V|$ and when $i>k$, $|BC[i]| \leq c_{vc} \cdot |BC[i-1]|$ for some positive $c_{vc}<1$ by Lemma~\ref{l:bcs}. Thus $|BC[k+m]|$ will reach zero for some $m=O(\log |V|)$, and when $|BC[k+m]|=0$, $C[i]$ is empty. 
\end{proof}

\pagebreakk\section{Computing the structure}
\jlabel{s:computing}

The structure is designed to be computed in time linear in its size. 

The set $C'[i]$ can be computed from $C[i-1]$ by looking from every $v' \in C[i]$ at those vertices in $G[i]$ that are at most $8^i$ from $v$. By Lemma~\ref{l:ciscover}, $C[i-1]$ is a $(8^{i-1},O(1))$ vertex cover, and since $G[i-1]$ has maximum degree $O(1)$ by Lemma~\ref{l:shorcutsize}, it only takes constant time to examine all $O(1)$ vertices in $G[i-1]$ that are at most $8^i$ from $v$. 

Note that when looking at an edge $v_1,v_2$ in a shortcut graph, we need to know the value of $\maxedge(p(v_1,v_2))$ in $G$; this information can easily be augmented on each shortcut edge and computed in constant time at the time of its construction.

\pagebreakk\section{Computing shortest paths and supporting dynamic changes}
\jlabel{s:dynamic}

One familiar with the work done in \cite{HD} will realize that, as in previous sections we described how to compute efficiently a hierarchy of SPHS's, we can directly apply the result of \cite{HD} and conclude that it is possible to construct a data-structure supporting shortest distance, shortest path queries and others.

However, one main achievement of this work is to be able to update our structure to allow queries on dynamically changing graphs. The queries we support are where vertices and edges can be added, deleted, or re-weighted. The issue with the proposed approach in \cite{HD} is that once the SPHS's are defined, these are used as input to construct a new data-structure globally. We would want to find a strategy where local changes are only perturbing the data structure locally. Therefore we will start by showing how to maintain the $C[i]$'s, and how to execute queries directly on the $G[i]$'s. 

\subsection[{Maintaining C[i]'s}.]{Maintaining $C[i]$'s}
\label{s:mci}
We assume that we have a graph $G=(V,E)$ for which we already computed our structure, that is we already have a set of $C[i]$, $C'[i]$, $E[i]$ and $G[i]$ for all $i$. Suppose we have an update (insertion, deletion and weight change); we always consider updates to be edge-related, and because we work on connected graphs, removing the last edge adjacent to a vertex implies the deletion of the vertex; while if an edge is added it must always have one of its endpoints already present in $G$. 

We denote the updated edge by $e=(v_1,v_2)$ and by $G^*=(V^*,E^*)$ the final graph after the update.

Remember that our algorithm to construct $C[i]$ proceeded, for each $i$, by considering all pairs of vertices from $G[i-1]$ whose shortest path in G has length in $[\frac{3}{4}\cdot8^i,8^{i}]$ and $\maxedge_G(p(v,v'))\leq 8^{i-1}$. We will assign pairs of vertices to a color to distinguish them, by extension we say that we color shortest paths, while we actually mean assign a color to its defining pair of vertices.

We color this collection of shortest paths as follows (See Figure~\ref{fig:10}): 
\begin{itemize}
    \item In red all the shortest paths that resulted in adding a vertex to $C'[i]$, and which contain at least one vertex outside of $B(v_1,2\cdot8^i) \cap B(v_2,2\cdot8^i)$.
    \item In black all the shortest paths that did not result in adding a vertex to $C'[i]$, and which contain at least one vertex outside of $B(v_1,2\cdot8^i) \cap B(v_2,2\cdot8^i)$.
    \item In blue all the remaining shortest paths of the collection.
\end{itemize}

\begin{figure}
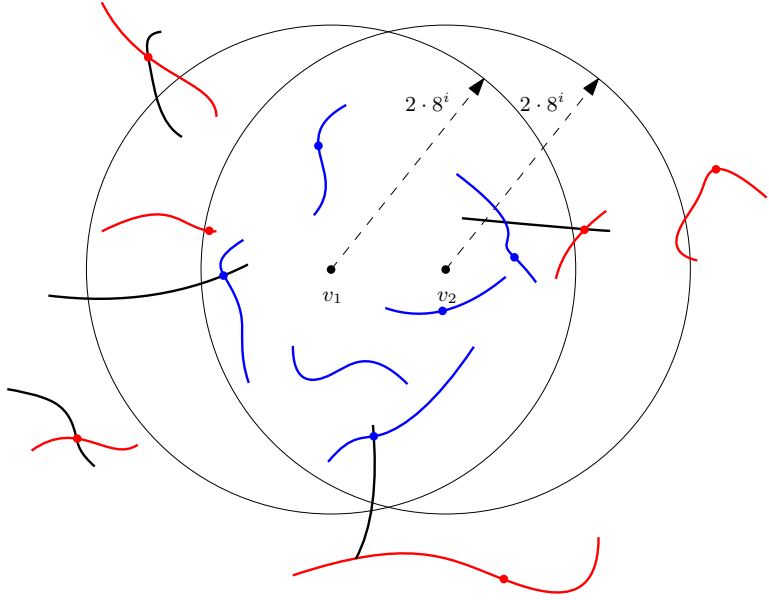

    \centering
    \finalfigfrompage{10}{0}
    \caption{Illustration of coloring of the paths and vertices in Section~\ref{s:mci}.}
    \label{fig:10}
\end{figure}

This coloring scheme depends on the updated edge and the actual order used by the algorithm. By extension, we say that a vertex of a $C'[i]$ is red or blue if it has been added by a corresponding shortest path. In case where either $v_1$ or $v_2$ was not in $G$, we only consider the ball centered at the endpoint that exists in both graphs; and all balls are $B_G$, i.e. defined based on the weight in $G$.

To process the update of edge $e$, we do as follows: we start from $C'[i]$, we remove all the blue vertices to obtain the set $C_{init}$, then we apply our incremental algorithm from Definition~\ref{d:construct} initializing $C'^*[i]$ to $C_{init}$, to obtain $C'^*[i]$ on all the shortest paths of $G^*$ that do contain a vertex of $B(v_1,2\cdot8^i)$ and $B(v_2,2\cdot8^i)$. $C^*[i]$ is then computed by (a) performing the same changes as in $C'^*[i]$ and (b) correcting the  edges $E[>i]$ if the updated edge $e$ has changed weight, was added or removed.

Our claim is as follows: we will show that if we construct $C^*[i]$'s and $C'^*[i]$'s directly using Definition~\ref{d:construct} on $G^*$, there exists an ordering of the pairs of vertices from $G^*[i-1]$ that will result in the exact same sets as the ones obtained by the update. This proves that our updating scheme results in \emph{valid} $C[i]$'s. The proof is inductive and it trivially holds for $i=-1$.

\begin{lemma}\jlabel{l:sameshortest}
All shortest paths with length smaller than $8^{i}$ and $\maxedge_G(p(v,v'))\leq 8^{i-1}$ of $G[i-1]$ that contain at least one vertex outside of $B(v_1,2\cdot8^i) \cup B(v_2,2\cdot8^i)$ are also shortest paths in $G^*[i]$.
\end{lemma}

\begin{proof}
Let $s$ be one of these shortest paths. Because $s$ has length smaller than $8^i$ and contains a vertex outside of $B(v_1,2\cdot8^i) \cup B(v_2,2\cdot8^i)$, we know that every vertex of $s$ is at distance at least $8^i$ from $v_1$ and $v_2$ or that $d_{G^*}(v_1,v_2) > 8^i$. So $s$ does not contain the edge $e$ and no path including $e$ from $v_1$ to $v_2$ has length smaller than $8^i$. 
\end{proof}
In particular, this implies that black and red sets of shortest paths from the construction of $C[i]$ will be considered when building $C^*[i]$. We will see below that reds will stay red (i.e. will still need to be covered by a vertex in $C^*[i]$), but that some black could need to be covered in $C^*[i]$ while they did not generate a vertex in $C[i]$.

A shortest path with length in range $[\frac{3}{4}\cdot8^i,8^{i}]$ and $\maxedge_G(p(v,v'))\leq 8^{i-1}$ is called a {\bf considered} shortest path.

\begin{lemma}\jlabel{l:outsidecovered}
All considered shortest paths of $G^*[i-1]$ that do not contain a vertex of $B(v_1,2\cdot8^i) \cup B(v_2,2\cdot8^i)$ are properly covered by $C_{init}$.
\end{lemma}

\begin{proof}
This is by construction: by Lemma~\ref{l:sameshortest}, the set of considered shortest path outside of the balls exists in both graphs, so if $C_{init}$ covers those paths in $C'[i]$, it also covers them in $C'^*[i]$. Note that all these shortest paths are red or black paths, and $C_{init}$ contains all red vertices of $C'[i]$. So we already know that all red paths strictly outside of the balls are covered. It remains to see that black shortest path strictly outside of the balls could not be covered by a blue vertex, as no blue vertex is outside of the balls by definition of blue.
\end{proof}

\begin{lemma}
There exists an ordering of the red shortest paths such that the set $C_{init}$ corresponds to the intermediate $C'^*[i]$ obtained by the method of Definition~\ref{d:construct} after we only considered red shortest paths.
\end{lemma}

\begin{proof}
If we apply the algorithm by considering red shortest paths in the same order as in the original $C[i]$, we know that each red path will require to add a new vertex to $C'^*[i]$, and $C_{init}$ is exactly the set of all those red vertices.
\end{proof}

\begin{lemma}
Applying the algorithm of Definition~\ref{d:construct}, starting with the set $C_{init}$ and applying the process to all considered shortest paths that have at least one vertex inside $B(v_1,2\cdot8^i) \cup B(v_2,2\cdot8^i)$ results in a valid $C^*[i]$.
\end{lemma}

\begin{proof}
It suffices to notice that the proposed algorithm is equivalent to starting with an empty $C'^*[i]$, then starting with all red shortest paths (whose corresponding vertices are those in $C_{init}$), then with all black shortest paths, and then the remaining ones. Note that pairs that did not exist in the original graph will need to be considered if $v_1$ or $v_2$ is a new vertex or if the weight change results in new candidates inside the ball. 

Notice that processing considered shortest paths outside of the balls is not necessary, as we already know that they are covered in $C_{init}$, and that all path that will add something to $C_{init}$ are path with at least one vertex inside the ball; which is exactly the set of path we consider.
\end{proof}

The astute reader will note that by visiting all considered shortest paths including a vertex of one of the balls, we will re-visit paths that are already covered in $C_{init}$, and that some path that were previously black could now generate a vertex. 

\begin{lemma}\jlabel{l:levelisohofone}
For each $i$, updating $G[i]$, $C[i]$, $C'[i]$, $E[i]$ to get $G^*[i]$, $C^*[i]$, $C'^*[i]$, $E^*[i]$ after insertion, deletion or re-weighting of an edge $e=(v_1,v_2)$ takes $O(1)$ time.
\end{lemma}

\begin{proof}
Our updating process consists in performing computations on all pairs of vertices of $G^*[i-1]$ whose shortest paths of length at most $8^i$ contains a vertex of $B(v_1,2\cdot8^i)$ or $B(v_2,2\cdot8^i)$. This implies that the vertices to consider are all in $B(v_1,3\cdot8^i)$ or $B(v_2,3\cdot8^i)$ or the extra vertex $v_1$ or $v_2$ in case the update added a new vertex.

Listing these vertices can be done by growing a ball from $v_1$ and $v_2$ in $G[i-1]$ (e.g. using Dijkstra's algorithm).

As $C[i]$'s are vertex covers, we know that there are $O(1)$ items to consider in these balls, $O(1)$ pairs of vertices to consider, implying that we are working on a problem of constant size. 

Updating $G[i]$ is a similar problem: edges in $G[i]$ are shortest paths in $G[i-1]$ (and therefore by induction in $G$). By Lemma~\ref{l:sameshortest}, all shortest paths with at least one vertex outside the balls are identical in both $G[i]$ and $G^*[i]$, so we only need to re-compute the shortcut graph for pairs of vertices $(v_4,v_4)$ for all $v_4,v_4$ in $C^*[i]$ inside the ball. This is again is a constant size problem. 
\end{proof}

\begin{lemma}
Processing an update takes $O(\log(U+|V|))$ time.
\end{lemma}

\begin{proof}
By Lemma~\ref{l:height}, we know that the height of the structure is at most $O(\log(U+|V|))$, and for each $i$ we process the elements in $O(1)$ time by Lemma~\ref{l:levelisohofone}.
\end{proof}

\begin{defn}
 Given a graph $G=(V,E)$, and vertices $v_1$ and $v_2$ of $V$, the \emph{funnel graph} $F_G(v_1,v_2)$ is $\bigcup_{i=0}^\infty G[i] \cap ( B(v_1,8^{i+1}) \cup B(v_2,8^{i+1}) )$. 
\end{defn}

Note that in general funnel graphs are not unique, they derive from any valid set of $C[i]$'s.

\begin{lemma} \jlabel{l:funnelisgood}{\bf Distance is preserved in funnel graphs. } For all $G=(V,E), v_1, v_2 \in V$ we have $d_G(v_1,v_2) = d_{F_G}(v_1,v_2)$.
\end{lemma}

\begin{proof}
By construction, each edge $e=(v_3,v_4)$ in any $G[i]$ has weight $d_{G[i]}(e) = d_G(v_3,v_4)$.

By contradiction, suppose the shortest path $d_{F_G}(v_1,v_2)$ in the funnel is not equal to $d_G(v_1,v_2)$ in $G$. 

Let $j$ be the maximum value such that $d_G(v_1,v_2) > 8^j$. We know that $G[j]$ contains at least one vertex $w \in V(p_G(v_1,v_2))$, as $C[j]$ is a vertex cover: either the shortest path is such that $\maxedge_G(p(v_1,v_2)) \leq 8^j$, or if there is a long edge in the shortest path, its vertices were also included in $G[j]$. Note that $d_G(v_1,w)\leq 8^{j+1}$ and $d_G(v_2,w)\leq 8^{j+1}$, therefore vertex $w$ is in the funnel graph. 

We now proceed by induction, from $v_1$ to $w$; let $j'$ be the maximum value such that $d_G(v_1,w) > 2\cdot8^{j'}$. There exists a vertex $w' \neq w$ on the shortest path from $v_1$ to $w$, that is in $G[j']$. In case where there are multiple candidates, we take the vertex closest to $v_1$, and we are certain that $d_G(v_1,w') \leq 8^{j'}$, and is therefore part of $G[j']$ and of $G[j'-1]$. That is because any subpath of length at least $8^{j'}$ is covered by a vertex, while the path we consider is twice as long. The shortest path from $v_1$ to $w'$ is the same in $G$ as in $F_G(v_1,v_2)$ by induction, while the shortest path from $w'$ to $w$ corresponds to edges in $G[j']$, as $w'$ and $w$ are both vertices of $C[j']$. This can be seen because $w$ is a vertex of $C[j]$, which is a subset of $C[j']$, as $j'\leq j$. The path from $w'$ to $w$ is therefore in the funnel $F_G(v_1,v_2)$, as the funnel contains the shortcut graph between all vertices in the ball $B(v_1,8^{j'+1})$. The base case is when $C[j"]$ contains $v_1$, which will occur at the latest when $8^{j"}$ is smaller than the first edge of the shortest path, implying the shortest path from $v_1$ to $w"$ is also in $C[j"]$. 

The symmetrical case from $w$ to $v_2$ is solved accordingly, concluding that we found in $F_G(v_1,v_2)$ a path with length exactly the same as in $G$.
\end{proof}

\begin{lemma}
Given $G=(V,E), v_1, v_2 \in V$, the size of the funnel graph $F_G(v_1,v_2)$ is $O(\log(U+|V|))$.
\end{lemma}

\begin{proof}
In every $G[i]$ the funnel contains the sub-graph within a constant number of balls of radius $8^{i+1}$. We know that the number of edges and vertices of $G[i]$ inside these balls is a constant in a graph of constant highway dimension. As we consider $O(\log(U+V))$ such graphs, the overall complexity is $O(\log(U+V))$.
\end{proof}

\begin{lemma}\jlabel{l:expandpath}
Given a shortest path in $F_G(v_1,v_2)$, we can find a shortest path in $G=(V,E)$ in $O(\log(U+|V|)+p)$ time, where $p$ is the number of edges in shortest path $p_G(v_1,v_2)$.
\end{lemma}

\begin{proof}
We expand to a shortest path in $G$ inductively: for each edge of some $G[i]$ in our structure, we expand it in $O(l)$ time to the corresponding edges in $G[i-1]$, where $l$ is the size of the shortest path corresponding to that edge in $G[i-1]$. Total expansion time is $O(\log(U+|V|) + p)$ where $p$ is the number of edges of the shortest path: at each level $i$ we either add $l$ edges to the path, or we go down one level to $i-1$. The induction ends when each edge corresponds to an edge of $G$.
\end{proof}

\begin{lemma}
Let $G$ be a connected graph with $|V|$ vertices, constant highway dimension, and ratio of smallest-to-largest edge $U$. Our data structure supports
shortest distance queries in time $O(\log (|V|+U))$ and shortest path queries whose result is a path with $p$ edges in time $O(p+\log (|V|+U))$. 
\end{lemma}

\begin{proof}
Let $v_1$ to $v_2$ be the vertices for which we want to compute the shortest path or the distance. We first compute a shortest path between $v_1$ and $v_2$ in $F_G(v_1,v_2)$. Note that in $F_G(v_1,v_2)$, shortest paths are not unique. But by Lemma~\ref{l:funnelisgood}, finding a shortest path in $F_G(v_1,v_2)$ will give a path of the same length as the unique shortest paths in $G$, and by Lemma~\ref{l:expandpath}, we can expand this shortest path in $O(\log(U+V) + p$ time to the unique shortest path $p_G(v_1,v_2)$.

To compute the shortest paths in the funnel $F_G(v_1,v_2)$, we proceed again by iterating on $i$ from $0$ to $O(\log(U+V))$. By induction, we already know the shortest path from $v_1$ to all vertices in $G[i-1]$, and we determine the shortest path from $v_1$ to all vertices in $G[i]$, limited to the ball $B(v_1, 8^{i+1})$. This is a constant-size problem. We do the same for $v_1$ and $v_2$ and as soon as a level contains a common vertex we maintain a best shortest path candidate.  

So in $O(\log(U+V))$ time we find a shortest path in $F_G(v_1,v_2)$ and can determine the weight of the shortest path, while $O(\log(U+V)) + p$ is required if we need to return the path by expanding it.
\end{proof}

\begin{theorem}
Our data structure can be computed in $O(|V| \log U)$ time and takes $O(|V| \log U)$ space. The data structure supports dynamic changes to the graph, edge insertions/deletions/relabeling in time $O(\log (|V|+U))$.
\end{theorem}

\bibliographystyle{abbrv}
\bibliography{bib}

\end{document}